\documentclass[letterpaper, 10 pt, conference]{ieeeconf}

\IEEEoverridecommandlockouts            

\usepackage{graphics} 
\usepackage{epsfig}
\usepackage{times} 
\usepackage{amsthm}
\usepackage{amsmath}
\usepackage{amssymb}
\usepackage{balance}
\usepackage{dsfont}
\usepackage{algorithm}
\usepackage{algpseudocode}
\usepackage{subfig}
\usepackage{setspace}
\usepackage{multirow}
\usepackage{enumerate}
\usepackage{tabularx}
\usepackage{balance}
\usepackage[noadjust]{cite}
\usepackage{thmtools}

\usepackage{bm}
\newtheorem{remark}{Remark}
\newtheorem{theorem}{Theorem}
\def\p(#1|#2){p(#1\,|\,#2)}

\let\Algorithm\algorithm
\renewcommand\algorithm[1][]{\Algorithm[#1]\setstretch{1.2}}
\renewcommand{\baselinestretch}{1.05}

\title{\LARGE \bf
Optimal Power Management of Battery Energy Storage Systems via Ensemble Kalman Inversion}

\author{Amir Farakhor$^{1}$, Iman Askari$^{1}$, Di Wu$^{2}$, and Huazhen Fang$^{1}$% <-this % stops a space
\thanks{$^{1}$A. Farakhor, I. Askari and H. Fang are with the Department of Mechanical Engineering, University of Kansas, Lawrence, KS, USA.
        Email: {\tt\small \{a.farakhor, askari, fang\}@ku.edu}}%
\thanks{$^{2}$D. Wu is with the Pacific Northwest National Laboratory, Richland, WA, USA.
        Email: {\tt\small di.wu@pnnl.gov}}%
}

\begin{document} 

\maketitle
\thispagestyle{empty}
\pagestyle{empty}

\begin{abstract}
Optimal power management of battery energy storage systems (BESS) is crucial for their safe and efficient operation. Numerical optimization techniques are frequently utilized to solve the optimal power management problems. However, these techniques often fall short of delivering real-time solutions for large-scale BESS due to their computational complexity. To address this issue, this paper proposes a computationally efficient approach. We introduce a new set of decision variables called {\em power-sharing ratios} corresponding to each cell, indicating their allocated power share from the output power demand. We then formulate an optimal power management problem to minimize the system-wide power losses while ensuring compliance with safety, balancing, and power supply-demand match constraints. To efficiently solve this problem, a parameterized control policy is designed and leveraged to transform the optimal power management problem into a parameter estimation problem. We then implement the ensemble Kalman inversion to estimate the optimal parameter set. The proposed approach significantly reduces computational requirements due to 1) the much lower dimensionality of the decision parameters and 2) the estimation treatment of the optimal power management problem. Finally, we conduct extensive simulations to validate the effectiveness of the proposed approach. The results show promise in accuracy and computation time compared with explored numerical optimization techniques.
\end{abstract}

\section{Introduction}
Battery energy storage systems (BESS) have found widespread use in various applications ranging from small portable electronics to large-scale battery packs for electric vehicles and grid energy storage \cite{TIE-AA-2005,JPS-DR-2015}. These BESS comprise many battery cells connected in series/parallel to deliver the required output voltage/capacity requirements. BESS power management commonly refers to power distribution among these constituent cells for the sake of multiple operational objectives. Conventional power management strategies mainly focus on state-of-charge (SoC) balancing to maximize the BESS utilization \cite{IECON-AB-2018}. However, recent studies have pointed out further potentials that lie in optimal power management, including temperature balancing and output voltage regulation \cite{IVPPC-BJV-2014}. Despite the evident advantages of optimal power management, its practical implementation presents a nontrivial challenge. The underlying optimization problem can quickly reach a formidable level of computational complexity with an increasing number of cells, rendering real-time execution infeasible. While different strategies have been explored in the literature to alleviate computational requirements, optimal power management remains hardly attainable for large-scale BESS. In this paper, we propose a novel perspective on optimal power management, leveraging Bayesian estimation to significantly reduce computation. Next, we will review the existing literature on optimal power management.
\subsection{Literature Review}
Optimal power management brings about several key operational features, including power loss minimization, cell balancing, and charge/discharge control. To fully harness these functionalities, the literature has devised various optimization problems. Conventionally, the focus has centered around SoC balancing to maximize the BESS utilization. For instance, the study in \cite{ECC-PM-2013} introduces a linear program aimed at achieving SoC balancing either in minimum time or with minimum power loss. This work is subsequently expanded in \cite{ITEC-RG-2015} to enable inter- and intra-module level SoC balancing. Further, in \cite{IFAC-NM-2012}, a convex optimization problem is formulated to achieve multiple objectives of cell balancing and power loss minimization. Building upon this work, the studies in \cite{ITSE-PC-2016, ITVT-CR-2019} extend this formulation to accommodate other BESS configurations. We have also tailored a convex optimization problem for a reconfigurable BESS in our previous study \cite{TTE-FA-2023}. The study in \cite{ITCST-CR-2021} also formulates a hierarchical model predictive control for optimal power management. %In the work presented in \cite{TIE-WZ-2022}, optimal power management is employed to facilitate user-defined charging strategies while considering constraints related to cell balancing and temperature management.

%\cite{TII-OQ-2018}

In hindsight, optimal power management approaches leverage the power of their underlying optimization problems to enhance BESS performance. These problems are typically solved through computationally intensive numerical techniques. However, computational demands of these methods become enormous when large-scale BESS presents large numbers of decision variables, thus hindering the practical implementation. 

Distributed control is a valuable approach to managing large-scale systems. This method involves distributing control tasks and computations among the constituent units or agents of the system. This approach provides improved efficiency and scalability, making it a useful choice for BESS power control. Recent studies in \cite{TPEL-MT-2016,ITSE-OQ-2018} treat individual cells as autonomous agents and employed the concept of distributed average consensus within networked multi-agent systems to develop algorithms for SoC balancing. However, these methods have limitations in optimizing critical metrics such as power losses, making them non-optimal. Our previous research in \cite{ACC-FA-2023} also proposes an innovative distributed optimal power management approach that is specifically designed for large-scale BESS.
\subsection{Summary of Contributions}
Despite efforts to reduce computational complexity, optimal power management for large-scale BESS applications still remains beyond the reach of existing methods. Towards overcoming the situation, our paper delivers the following contributions:

\begin{enumerate}[1)]
	\item We introduce power-sharing ratios for the cells which indicate their power quota from the demanded output power. This facilitates formulating an optimal power management problem, leading to two benefits: firstly, the optimal power-sharing ratios depend solely on current cell conditions such as SoC, temperature, and internal resistance. This allows us to propose a parameterized control policy to reduce the number of decision variables remarkably. Secondly, it streamlines the Bayesian estimation treatment of the optimal power management problem by limiting the search space of control variables.
	\item We formulate and convert an optimal power management problem into a Bayesian parameter estimation problem. Avoiding tedious numerical optimization, we instead perform estimation to enable optimal power management. We specifically leverage ensemble Kalman inversion to estimate the optimal parameter set for the proposed control policy. The estimation framework substantially reduces computation time, making optimal power management potentially feasible and scalable for real-world applications.
\end{enumerate}

We validate our approach through extensive simulations, demonstrating its efficacy in terms of accuracy and computation time.
\section{Optimal Power Management of BESS}
This section elaborates on the BESS circuit structure, its electro-thermal modeling, and optimal power management formulation. 
\subsection{Circuit Structure}
Fig.~\ref{FIG_1} illustrates the circuit structure of the considered BESS. It comprises $n$ battery cells. Each cell is equipped with its corresponding DC/DC converter, which can be interconnected in series or parallel configurations to supply the load under specific voltage/capacity ratings. The DC/DC converters allow bidirectional power processing to charge/discharge the cells with regulated power. This brings about cell-level power control for the considered BESS. This paper leverages this feature to enable cell balancing and power loss minimization. It is worth mentioning that the proposed approach in this work extends to other circuit structures capable of cell-level power control. The following subsection provides an electro-thermal model for the BESS to pave the path for the problem formulation. 

%Our previous studies in \cite{IECON-FA-2021,TTE-FA-2023} consider a reconfigurable connection among the converters. However, we consider a fixed configuration here to focus solely on optimal power management.

\begin{figure}[!t]\centering
	\includegraphics[width=8.5cm]{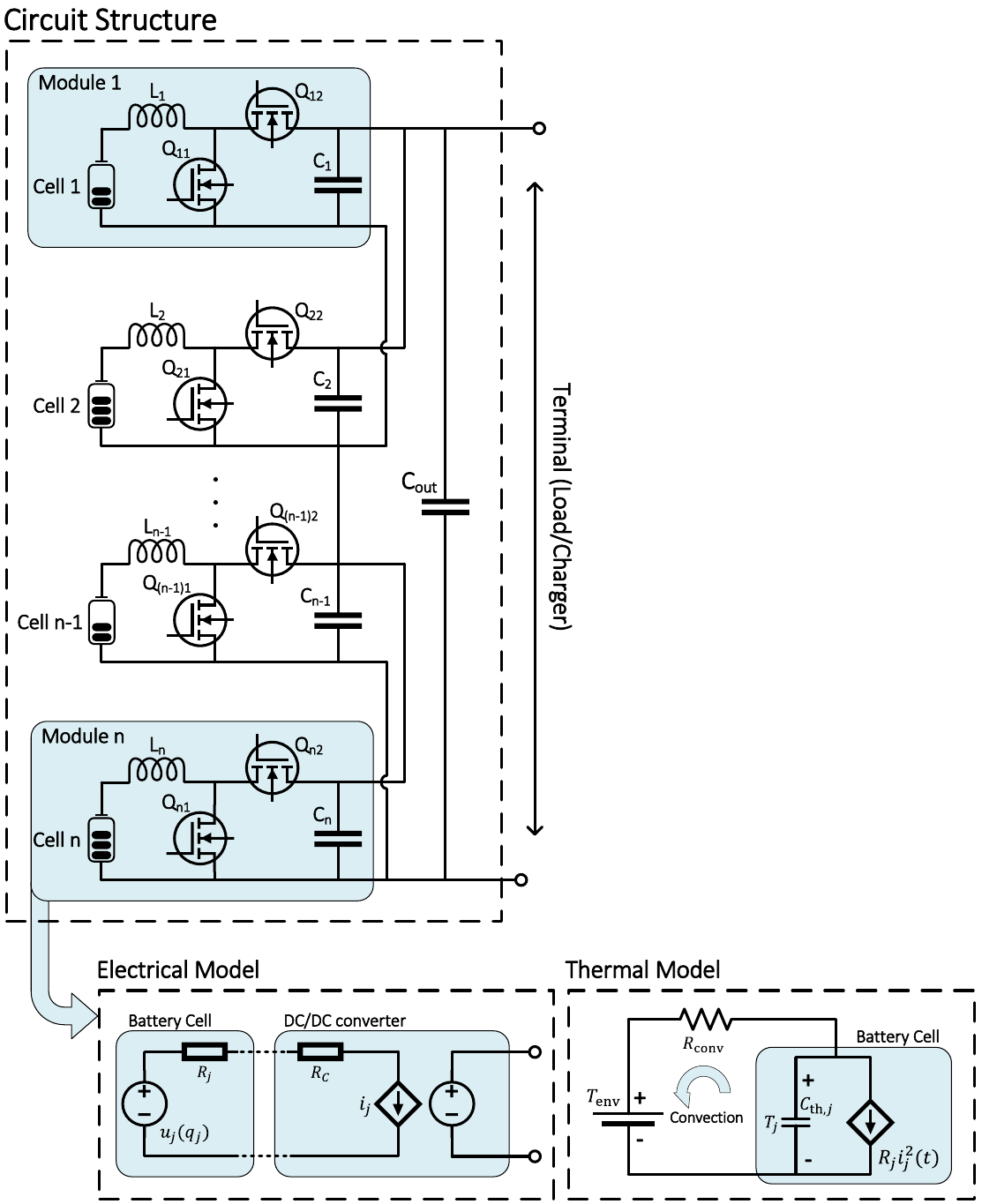}
	\caption{The circuit structure of large-scale BESS.}
	\label{FIG_1}
\end{figure}
\subsection{Electro-thermal Modeling}
We begin by characterizing the electrical dynamics of the considered BESS. Each cell is described by the Rint model, which encompasses a SoC-dependent open-circuit voltage (OCV) and a series resistance \cite{ENE-HH-2011}. Further, we model the DC/DC converters using an ideal DC transform and a series resistance to capture their inherent inefficiencies. The electrical model is illustrated in Fig.~\ref{FIG_1}. For cell $j$, the governing equations are given by
\begin{subequations}
	\begin{align}
		\dot{q}_j(t)&=-\frac{1}{Q_j}i_{j}(t), \label{SoC_Dynamics}\\
		v_j(t)&=u_j(q_j(t))-R_ji_{j}(t),
	\end{align}
\end{subequations}
where $v_j$, $u_j$, $R_j$, $i_j$, $Q_j$, $q_j$ are the terminal voltage, OCV, internal resistance, current, capacity, and SoC, respectively. The cell's internal charging/discharging power $P_{b_j}$ is defined as follows:
\begin{align}
	P_{b_j}=u_j(q_j(t))i_{j}(t).
\end{align}
Further, the cell’s and its corresponding converter’s power losses are expressed by
\begin{align}
	P_{L_j}=(R_j + R_C)i_{j}^2(t).
	\label{Loss}
\end{align}
For the cells, let us introduce the power-sharing ratios. This quantity indicates the portion of power allocated to each cell from the total output power demand. They are defined as
\begin{align}
	\mu_j=\frac{P_{b_j}}{|P_{\textrm{out}}|},
	\label{Util}
\end{align}
where $\mu_j$ and $P_{\textrm{out}}$ are the power-sharing ratio and output power demand, respectively. Note that we consider the absolute value of $P_{\textrm{out}}$ in \eqref{Util} to accommodate both charging and discharging cycles. Proceeding forward, we rewrite \eqref{SoC_Dynamics} and \eqref{Loss} to represent them in terms of $\mu_j$. The new expressions are as follows:
\begin{subequations}
	\begin{align}
		\dot{q}_j(t)&=-\frac{\mu_j(t)P_{\textrm{out}}(t)}{Q_ju_j(q_j(t))}, \label{SoC_Dynamics_Util}\\
		P_{L_j}&=\frac{(R_j + R_C)\mu_j^2(t)P_{\textrm{out}}^2(t)}{u_j^2(q_j(t))}.
	\end{align}
	\label{Elec_Dynamic}
\end{subequations}

Next, we describe the cells' thermal behavior using a lumped thermal circuit. The temperature of the cells is influenced by heat generation due to internal power losses, and heat dissipation due to environmental convection. For cell $j$, the governing dynamic equation is
\begin{align}
	C_{\textrm{th},j}\dot{T}_j(t) = \frac{R_j\mu_j^2P_{\textrm{out}}^2(t)}{u_j^2(q_j(t))} - (T_j(t) - T_{\textrm{env}})/R_{\textrm{conv}},
	\label{Temp_Dynamic}
\end{align}
where $C_{\textrm{th},j}$, $T_j$, $T_{\textrm{env}}$, and $R_{\textrm{conv}}$ are the thermal capacitance, cell's temperature, environmental temperature, and convective thermal resistance, respectively. The collective electro-thermal model, described in \eqref{Elec_Dynamic} and \eqref{Temp_Dynamic}, strikes an appropriate balance between expressiveness and computational tractability to pave the path toward optimal power management formulation.
\subsection{Problem Formulation}
Here, we introduce the proposed optimal power management approach based on the power-sharing ratios. The main objective of our approach is to achieve a power-loss-minimized BESS operation while adhering to safety, balancing, and power supply-demand match requirements. The formulation in the sequel is represented in discrete time after discretizing the electro-thermal model via the forward Euler method.

To begin with, we express the BESS power losses at time instant $t$ as follows: 
\begin{equation}
	J[t] = \sum_{j=1}^n \frac{(R_j + R_C)\mu_j^2[t]P_{\textrm{out}}^2[t]}{u_j^2(q_j[t])}.
\end{equation}
To ensure the safe operation of the BESS, we impose the following constraints:
\begin{subequations}
	\begin{align}
		i_{j}^\textrm{min}&\leq i_j\leq i_{j}^\textrm{max}, \label{SafetyConstraint}\\
		q_j^\textrm{min}&\leq q_j\leq q_j^\textrm{max}, \label{SoCConstraint}
	\end{align}
\end{subequations}
where $i_{j}^\textrm{min/max}$ and $q_j^\textrm{min/max}$ are the lower/upper safety bounds for the cells' current and SoC. Note that we can express \eqref{SafetyConstraint} in terms of the power-sharing ratios as follows: 
\begin{equation}
	\frac{u_j(q_j[t])}{\bigl|P_{\textrm{out}}[t]\bigr|}i_{j}^\textrm{min}\leq \mu_j[t]\leq \frac{u_j(q_j[t])}{\bigl|P_{\textrm{out}}[t]\bigr|}i_{j}^\textrm{max}.
	\label{SafetyConstraint-Util}
\end{equation}
To promote balanced cell use, we further enforce constraints on BESS operation as follows:
\begin{subequations}
	\begin{align}
		\bigl|q_j[t]-q_{\textrm{avg}}[t]\bigr|& \leq \Delta q, \label{SoCBalancing} \\
		\bigl|T_j[t]-T_{\textrm{avg}}[t]\bigr|& \leq \Delta T, \label{TempBalancing}
	\end{align}
	\label{BalancingConstraint}
\end{subequations}
\hspace{-5pt}where $q_{\textrm{avg}}[t]$ and $T_{\textrm{avg}}[t]$ represent the average SoC and temperature, and $\Delta q$ and $\Delta T$ are the allowed deviation of the cells' SoC and temperature from their respective average values. The power supply-demand balance can also be guaranteed using the following constraint:
\begin{equation}
	\sum_{j=1}^n \left(\mu_j[t] - \frac{(R_j+R_C)\mu_j^2[t]\bigl|P_{\textrm{out}}[t]\bigr|}{u_j^2(q_j[t])}\right) = 1.
	\label{PowerBalanceConstraint}
\end{equation}
Having laid out the power loss expression and the constraints, we can now proceed to formulate the optimal power management problem over a receding horizon as follows:
\begin{equation}
	\begin{aligned}
		\min_{\mu_{j}, j=1,...,n} \quad & \sum_{t=k}^{k+H}J[t],\\
		\textrm{s.t.} \quad & \eqref{SoC_Dynamics_Util},\eqref{Temp_Dynamic},\eqref{SoCConstraint},\eqref{SafetyConstraint-Util},\eqref{BalancingConstraint},\eqref{PowerBalanceConstraint},
	\label{Optim-1}
	\end{aligned}
\end{equation}
where $H$ is the horizon length. Note that the optimization problem stated in \eqref{Optim-1} is both non-linear and non-convex, and requires optimization of $3nH - n$ variables. While it may be feasible to manage its computation for small $n$, it can quickly become extremely challenging for large-scale BESS. Therefore, we aim to develop a computationally efficient solution for \eqref{Optim-1} in the next section.
\section{Efficient Optimal Power Management Approach}
This section introduces a parameterized control policy for the optimal power management problem and then translates it into a Bayesian parameter estimation problem. Finally, we leverage the ensemble Kalman inversion for parameter estimation.

\subsection{Input Parameterization}
As previously discussed, the computational complexity of \eqref{Optim-1} hinders the adoption of optimal power management for large-scale BESS. To alleviate this issue, we propose a parameterized control policy based on the BESS operational conditions. By doing so, we intend to map a high-dimensional control space, power-sharing ratios, into a low-dimensional parameter space. Before introducing the control policy, we explain the rationale behind this design with an example. Let us consider a BESS with uniform SoC distribution except for one cell with a relatively lower SoC value. To achieve a balance, we require this cell to be allocated more/less power in charging/discharging cycles. We also need analogous behavior for the cells’ temperatures and internal resistances. We thus parameterize the control policy as follows to generate such behaviors: 

\begin{equation}
	\mu_j[k] = \begin{bmatrix}\theta_1 & \theta_2 & 1-\theta_1-\theta_2\end{bmatrix}^\top \begin{bmatrix}l_{q, j} \\ l_{T, j} \\ l_{R, j} \end{bmatrix},
	\label{Form-1}
\end{equation}
where $\theta_1$ and $\theta_2$ are the parameters; $l_{q, j}$, $l_{T, j}$, and $l_{R, j}$ are also defined as follows:
\begin{subequations}
	\begin{align}
		l_{q,j} &=\begin{cases}\sum_{i=1}^{n}\left(\frac{q_j}{q_i}\right)^{-\beta_1} & P_{\textrm{out}} \geq 0 \\ \sum_{i=1}^{n}\left(\frac{q_i}{q_j}\right)^{-\beta_1} & P_{\textrm{out}} < 0\end{cases}, \label{q-Function}\\
		l_{T,j} &= \sum_{i=1}^{n}\left(\frac{T_j}{T_i}\right)^{-\beta_2}, \label{T-Function} \\
		l_{R,j} &= \left(\sum_{i=1}^{n}\frac{R_j}{R_i}\right)^{-1}, \label{R-Function}
	\end{align}
	\label{Functions}
\end{subequations}
\hspace{-4pt}where $\beta_1$ and $\beta_2$ are the hyperparameters. These hyperparameters play a crucial role in the SoC and temperature balancing performance, as they amplify the impact of imbalance on the power-sharing ratios. Specifically, higher values of $\beta$ lead to greater differentiation in power allocation among the cells. One can select $\beta_1$ and $\beta_2$ based on BESS specifications.

\begin{remark}
The control policy proposed in \eqref{Form-1} maintains optimality regarding power loss minimization. This is because \eqref{R-Function} is indeed the solution of \eqref{Optim-1} when the safety and balancing constraints \eqref{SoCConstraint}, \eqref{SafetyConstraint-Util}, and \eqref{BalancingConstraint} are not active. By setting $\theta_1=\theta_2=0$, the designed control policy reflects the absence of balancing constraints. Thus, the resultant power-sharing ratios achieve power-loss-minimizing operation.
\end{remark}

With the above control policy parameterization, the number of optimization variables reduce to $2nH+2$, in contrast to $3nH-n$ in the original problem. We will make further effort next to find out the best optimization variables with fast computation.
\subsection{Bayesian Estimation Formulation}
We leverage the designed control policy in \eqref{Form-1} to translate \eqref{Optim-1} into a Bayesian parameter estimation problem. The estimation framework centers around utilizing the cost function (power losses) and the constraints (safety, balancing, and power conservation) to collectively guide the inference of optimal parameter set. The primary objective of this transformation is twofold. Firstly, it provides us access to sophisticated tools developed in the field of estimation. Secondly, it leads to a significant reduction in computation. We begin the transformation by introducing a virtual dynamic system. While doing so, we use simplified notations for compact representations and adopt the convention of using boldface lowercase letters for vectors and boldface uppercase letters for matrices. The virtual dynamics are expressed as follows:
\begin{equation}
	\begin{cases}
		\boldsymbol{x}[t+1] = \boldsymbol{f}\bigl(\boldsymbol{x}[t], \boldsymbol{\pi}_{\boldsymbol{\theta}}(\boldsymbol{x}[t])\bigr), \\
		y[t] = h\bigl(\boldsymbol{x}[t], \boldsymbol{\pi}_{\boldsymbol{\theta}}(\boldsymbol{x}[t])\bigr) + v[t], \\
	\end{cases}
	\label{VirtualDynamics}
\end{equation}
where $\boldsymbol{x}[t]$, $\boldsymbol{\mu}[t]$, and $\boldsymbol{\theta}$ collect the cells’ SoC and temperature values, power-sharing ratios, and parameters as follows:
\begin{align}
	\boldsymbol{x}[t] &= \begin{bmatrix}q_1[t] & \dots & q_n[t] & T_1[t] & \dots & T_n[t]\end{bmatrix}^\top, \nonumber\\
	\boldsymbol{\mu}[t] &= \begin{bmatrix}\mu_1[t] & \dots & \mu_n[t]\end{bmatrix}^\top, \nonumber\\
	\boldsymbol{\theta} &= \begin{bmatrix}\theta_1 & \theta_2 & 1 - \theta_1 - \theta_2\end{bmatrix}^\top. \nonumber
\end{align}
Further, $\boldsymbol{f}$ denotes the discretized dynamic equations in \eqref{SoC_Dynamics_Util} and \eqref{Temp_Dynamic} for all cells; $\boldsymbol{\mu}[t]=\boldsymbol{\pi}_{\boldsymbol{\theta}}(\boldsymbol{x}[t])$ represents the parameterized control policy in \eqref{Form-1}. The virtual measurement $y[t]$ is characterized by
\begin{align}
	h\bigl(\boldsymbol{x}[t], \boldsymbol{\pi}_{\boldsymbol{\theta}}(\boldsymbol{x}[t])\bigr) = J[t] + \psi\bigl(g\left(\boldsymbol{x}[t], \boldsymbol{\pi}_{\boldsymbol{\theta}}(\boldsymbol{x}[t])\right)\bigr), 
\end{align}
where $g[t]$ gathers the constraints in \eqref{SoCConstraint}, \eqref{SafetyConstraint-Util}, \eqref{BalancingConstraint}, \eqref{PowerBalanceConstraint} for all cells at time $t$. Also, $\psi(\cdot)$ is a barrier function to measure the constraint satisfaction as follows:
\begin{align}
	\psi(x) = \begin{cases} 0 & x \leq 0 \\ \infty & x > 0\end{cases}.
	\label{Barrier}
\end{align}
In this framework, we relax the virtual measurement $y[t]$ to be stochastic by introducing a random noise variable $v[t]$. 

We highlight that the virtual system replicates the dynamics of the original system as in \eqref{Optim-1}. The only difference is that the hard constraints in \eqref{Optim-1} are relaxed to a soft constraint and added as a penalty term in the virtual measurements. Furthermore, it is crucial for the behavior of virtual measurements to mirror the characteristics of the optimization problem in \eqref{Optim-1}. This necessitates that the virtual measurements conform to the form where $y[t] = 0$ holds for $t=k,\dots,k+H$. Requiring $y[t] = 0$ drives the BESS operation close to the ideal case characterized by negligible power losses and constraint satisfaction.

We are now ready to perform parameter estimation given the virtual dynamics in \eqref{VirtualDynamics}. This involves the consideration of the posterior distribution $\p(\boldsymbol{\theta} | \boldsymbol{y}=\boldsymbol{0})$ with the concatenated measurements from $t=k,\dots,k+H$, i.e.,
\begin{equation}
    \boldsymbol{y}=\begin{bmatrix}y[k] & \dots & y[k+H]\end{bmatrix}^\top.
    \nonumber
\end{equation}
To estimate $\boldsymbol{\theta}$, we apply the maximum a posteriori (MAP) estimation:
\begin{equation}
	\boldsymbol{\hat{\theta}}^*= \textrm{arg}\max_{\boldsymbol{\theta}} \, \, \log \p(\boldsymbol{\theta} | \boldsymbol{y}=\boldsymbol{0}).
	\label{MAP}
\end{equation}

\begin{theorem}
Assume that $p(\boldsymbol{\theta})\sim \mathcal{N}(\boldsymbol{\bar{\theta}}, \boldsymbol{\Sigma}
^{\boldsymbol{\theta}})$ and that $p(v)\sim\mathcal{N}(0, R)$ in \eqref{VirtualDynamics} is white noise.
%\begin{equation}
%	p(\boldsymbol{\theta})\sim \mathcal{N}(\boldsymbol{\bar{\theta}}, \boldsymbol{\Sigma}
%^{\boldsymbol{\theta}}), \, \, p(v)\sim\mathcal{N}(0, R),
%	\nonumber
%\end{equation}
Then, the problems in \eqref{Optim-1} and \eqref{MAP} share the same optima when $\boldsymbol{\theta}$ has a noninformative prior, i.e., $\boldsymbol{\Sigma}^{\boldsymbol{\theta}} \rightarrow \infty$.
\end{theorem}

\begin{proof}
Using the Bayes' rule and the Markovian property of \eqref{VirtualDynamics}, we have
\begin{equation}
	\p(\boldsymbol{\theta} | \boldsymbol{y}=\boldsymbol{0}) \propto \prod_{t=k}^{k+H} \p(y[t]=0 | \boldsymbol{\theta})p(\boldsymbol{\theta}),
	\nonumber
\end{equation}
which implies
\begin{equation}
	\log \p(\boldsymbol{\theta} | \boldsymbol{y}=\boldsymbol{0}) = \sum_{t=k}^{k+H} \log \p(y[t]=0 | \boldsymbol{\theta}) + \log p(\boldsymbol{\theta}).
	\nonumber
\end{equation}
Because $\p(y[t]=0 | \boldsymbol{\theta})\propto p(v[t])\sim\mathcal{N}(0, R)$, we get
\begin{equation}
	\log \p(y[t]=0 | \boldsymbol{\theta}) \propto \frac{1}{R}h^2\left(\boldsymbol{x}[t], \boldsymbol{\pi}_{\boldsymbol{\theta}}(\boldsymbol{x}[t])\right).
	\nonumber
\end{equation}
Further, $\log p(\boldsymbol{\theta}) \rightarrow 0$ as $\boldsymbol{\Sigma}^{\boldsymbol{\theta}} \rightarrow \infty$. Putting together the above, we see that the cost function in \eqref{MAP} is the scaled opposite of the cost function and soft constraints in \eqref{Optim-1}. The theorem is thus proven.
\end{proof}
Theorem 1 suggests the equivalence between the estimation problem in \eqref{MAP} and the problem in \eqref{Optim-1} under mild conditions. However, finding an analytical solution to \eqref{MAP} is not possible, so we develop approximate solutions to \eqref{MAP}. It is important to note that the noninformative prior assumption for $\p(\boldsymbol{\theta})$ in Theorem 1 implies that the initial prior covariance should be large enough to enable the exploration of the entire search space. Next, we employ the ensemble Kalman inversion to approximately solve \eqref{MAP}.
\subsection{Parameter Estimation via Ensemble Kalman Inversion}
Tracing to ensemble Kalman filtering, the ensemble Kalman inversion method enables sampling-based approximation of a target posterior distribution \cite{IP-NK-2019}. It is effective in finding the unknown parameters of a complex system given measurement data. Its computation is derivative-free and fast by harnessing the power of sampling. Specifically, the method iteratively approximates $\p(\boldsymbol{\theta}_{\ell+1} | {\bm y})$ through
\begin{align} \label{Param-Bayes}
\p(\boldsymbol{\theta}_{\ell+1} | {\bm y}) \propto \p({\bm y} | \boldsymbol{\theta}_\ell)^{\alpha_\ell} p(\boldsymbol{\theta}_\ell),
\end{align}
where $\ell$ is the iteration index and $0\leq \alpha_{\ell} \leq 1$ is a varying hyperparameter used to temper the likelihood. As will be seen later, $\alpha_\ell$ will play the role of a stepsize parameter in the sampling-based update procedure \cite{MWR-PH-2016}. We consider that $p(\boldsymbol{\theta}_{\ell}) \sim \mathcal{N}(\boldsymbol{\bar \theta}_{\ell}, {\bm \Sigma}_l^{\boldsymbol{\theta}})$. Note that $\p({\bm y} | \boldsymbol{\theta}_{\ell})^{\alpha_{\ell}}$ has covariance $\alpha_{\ell}^{-1}\bm{R}$, where $\boldsymbol{R} = \mathrm{diag}(R, \ldots, R)$. We assume
\begin{align}\label{Joint-Gauss}
p \left( \begin{bmatrix}
\boldsymbol{\theta}_{\ell} \\ \bm{y}
\end{bmatrix}\right) \sim 
 \mathcal{N}\left(\begin{bmatrix}
 \boldsymbol{\bar \theta}_{\ell}\\ \bm{\bar y}_{\ell}
\end{bmatrix}, \begin{bmatrix}
\bm{\Sigma}_{\ell}^{\boldsymbol{\theta}} & \bm{\Sigma}_{\ell}^{\boldsymbol{\theta}\bm{y}} \\ \left(\bm{\Sigma}_{\ell}^{\boldsymbol{\theta} \bm{y}} \right)^\top & \bm{\Sigma}_{\ell}^{\bm{y}} +{\alpha_{\ell}}^{-1}\bm{R}
\end{bmatrix}\right).
\end{align}
It follows from \eqref{Joint-Gauss} that
\begin{align}\label{conditional-dist}
\p( \boldsymbol{\theta}_{\ell+1} | \bm{y}) \sim \mathcal{N} \left(\boldsymbol{\bar \theta}_{\ell+1} , \bm{\Sigma}_{\ell+1}^{\boldsymbol{\theta}}\right),
\end{align}
where 
\begin{subequations}\label{KF-update}
\begin{align}\label{KF-update-a}
\boldsymbol{\bar \theta}_{\ell+1} &= \boldsymbol{\bar \theta}_{\ell} +\bm{\Sigma}_{\ell}^{\boldsymbol{\theta} \bm{y}} \left( \bm{\Sigma}_{\ell}^{\bm{y}} + {\alpha_{\ell}}^{-1}\bm{R} \right)^{-1} \left(\bm{y} -\bm{\bar y}_{\ell} \right),\\ \label{KF-update-b}
\bm{\Sigma}_{\ell+1}^{\boldsymbol{\theta}} &= \bm{\Sigma}_{\ell}^{\boldsymbol{\theta}} +\bm{\Sigma}_{\ell}^{\boldsymbol{\theta} \bm{y}} \left( \bm{\Sigma}_{\ell}^{\bm{y}} +{\alpha_{\ell}}^{-1}\bm{R} \right)^{-1} \left(\bm{\Sigma}_{\ell}^{\boldsymbol{\theta} \bm{y}} \right)^\top.
\end{align}
\end{subequations}

Based on \eqref{conditional-dist}-\eqref{KF-update}, we can perform sampling-based approximation and update the posterior distribution. At iteration $\ell$, we have $\boldsymbol{\theta}_{\ell}^{(i)} \sim \mathcal{N}(\boldsymbol{\theta}_{\ell}, \bm{\Sigma}_{l}^{\boldsymbol{\theta}})$ for $i = 1, \ldots, N$, and then generate the sampled trajectories for the state for $k \leq t \leq k+H$ by
\begin{align}\label{Traj-Gen}
\bm {x}^{(i)}_{\ell}[t+1] = \boldsymbol{f}\left(\boldsymbol{x}^{(i)}_{\ell}[t], \boldsymbol{\pi}_{\boldsymbol{\theta}_{\ell}^{(i)}}(\boldsymbol{x}^{(i)}_{\ell}[t])\right), \quad i = 1, \ldots, N.
\end{align}
Then, the samples to approximate $\p(\boldsymbol{y}|\boldsymbol{\theta}_\ell)$ can be computed by
\begin{align}\label{Meas-Gen}
y_{\ell}^{(i)}[t] = h\left(\boldsymbol{x}^{(i)}_{\ell}[t], \boldsymbol{\pi}_{\boldsymbol{\theta}_{\ell}^{(i)}}(\boldsymbol{x}^{(i)}_{\ell}[t])\right) + v_{\ell}^{(i)}[t],
\end{align}
where $v_{\ell}^{(i)}[t] \sim \mathcal{N}(0, R)$. We now can calculate the means and covariances as below:
\begin{subequations}
	\begin{align}
		\boldsymbol{\bar \theta}_{\ell} &= \frac{1}{N}\sum_{i=1}^N \boldsymbol{\theta}_{\ell}^{(i)}, \quad \bm{\bar y}_{\ell} = \frac{1}{N}\sum_{i=1}^N \bm{y}_{\ell}^{(i)}, \label{Sample-Mean}\\
		\bm{\Sigma}^{\boldsymbol{\theta}}_{\ell} &= \frac{1}{N-1}\sum_{i=1}^N \left( \boldsymbol{\theta}_{\ell}^{(i)} -  \boldsymbol{\bar \theta}_{\ell}\right)\left( \boldsymbol{\theta}_{\ell}^{(i)} -  \boldsymbol{\bar \theta}_{\ell}\right)^{\top}, \label{Sample-Cov}\\
		\bm{\Sigma}^{\bm{y}}_{\ell} &= \frac{1}{N-1}\sum_{i=1}^N \left( {\bm y}_{\ell}^{(i)} -  \bm{\bar y}_{\ell} \right)\left( {\bm y}_{\ell}^{(i)} -  \bm{\bar y}_{\ell} \right)^{\top}, \\
		\bm{\Sigma}^{\boldsymbol{\theta}\bm{y}}_{\ell} &= \frac{1}{N-1}\sum_{i=1}^N \left( \boldsymbol{\theta}_{\ell}^{(i)} -  \boldsymbol{\bar \theta}_{\ell}\right)\left( {\bm y}_{\ell}^{(i)} -  \bm{\bar y}_{\ell} \right)^{\top}. \label{Sample-CrossCov}
	\end{align}
\end{subequations}
By \eqref{KF-update-a}, the samples to approximate $\p(\boldsymbol{\theta}_{\ell+1} | \boldsymbol{y})$ are 
\begin{align}\label{sample-update}
\boldsymbol{\theta}^{(i)}_{\ell+1} = \boldsymbol{\theta}^{(i)}_{\ell} +\bm{\Sigma}_{\ell}^{\boldsymbol{\theta} \bm{y}} \left( \bm{\Sigma}_{\ell}^{\bm{y}} + {\alpha_{\ell}}^{-1}\bm{R} \right)^{-1} \left(\bm{y} -\bm{ y}_{\ell}^{(i)} \right).
\end{align}
Generally, $\alpha_{\ell}$ takes a small value in early iterations, and increases towards 1 as the iteration moves forward. One can also use a bisection search algorithm to identify the best $\alpha_{\ell}$ for each iteration \cite{StatLetters-SD-2022}. Then, we have
\begin{align}\label{Samples-Posterior}
\boldsymbol{\theta}_{\ell+1} = \frac{1}{N}\sum_{i=1}^N \boldsymbol{\theta}^{(i)}_{\ell+1}.
\end{align}
The iteration procedure repeats itself until convergence when the difference in two consecutive iterations is less than a pre-specified tolerance $\epsilon$. This is the ensemble Kalman inversion method. The sampling-based iteration yields approximations of the posterior distribution of the unknown parameters, which can also be viewed as a search within the parameter space. The method is useful in dealing with complex nonlinear relationships from the parameters to the measurements and offers much faster computation than gradient-based optimization when dealing with many complex problems. Algorithm~\ref{EnKI-OPM} shows a summary of optimal power management based on the method. 

\begin{algorithm}[t]
\fontsize{9.2}{10}
  \caption{{\tt EnKI-OPM} Ensemble Kalman inversion-based optimal power management} \label{NMPC-EnKS}
  \begin{algorithmic}[1]
\State Set up the optimization problem as in~\eqref{Optim-1} 

\State Convert the problem to a Bayesian parameter estimation problem of $\p(\boldsymbol{\theta} | \bm{y} = \bm{0})$
 
\State Initialize by setting  $\ell = 0$, $\boldsymbol{\bar \theta}_0$ , $\bm{\Sigma}_0^{\boldsymbol{\theta}}$, and $\Delta \boldsymbol{\theta}_0$
\State Draw samples $\boldsymbol{\theta}_0^{(i)} \sim  \mathcal{N}(\boldsymbol{\bar \theta}_0, \bm{\Sigma}_0^{\boldsymbol{\theta}})$, for $i=1, \ldots, N$
\While{$\Delta \boldsymbol{\theta}_\ell \geq \epsilon$} 
\For{$t = k,\ldots, k+H$}
\State Generate $\bm{x}_{\ell}^{(i)}[t]$ using~\eqref{Traj-Gen}
\State Generate $y_{\ell}^{(i)}[t]$ as in~\eqref{Meas-Gen}
\EndFor
\State Construct $\bm{y}_{\ell}^{(i)}=\begin{bmatrix}y_{\ell}^{(i)}[k] & \dots & y_{\ell}^{(i)}[k+H]\end{bmatrix}^\top$

\State Compute $\boldsymbol{\bar \theta}_{\ell}$ and $\bm{\bar y}_{\ell}$ via~\eqref{Sample-Mean}
\State Compute $\bm{\Sigma}^{\boldsymbol{\theta}}_{\ell}$, $\bm{\Sigma}^{\bm{y}}_{\ell}$, and $\bm{\Sigma}^{\boldsymbol{\theta} \bm{y}}_{\ell}$ via~\eqref{Sample-Cov} -~\eqref{Sample-CrossCov}
\State Specify $\alpha_{\ell}$ via bisection search
\State Update $\boldsymbol{\theta}_{\ell+1}^{(i)}$ via~\eqref{sample-update}
\State Compute $\boldsymbol{\bar \theta}_{\ell+1}$ using~\eqref{Samples-Posterior}
\State Compute criteria $\Delta \boldsymbol{\theta}_{\ell+1} \leftarrow \lVert\boldsymbol{\bar \theta}_{\ell+1} - \boldsymbol{\bar \theta}_{\ell}\rVert$
\State Set $\ell \leftarrow \ell + 1$
\EndWhile

\State Extract $\boldsymbol{\hat \theta}^\ast[k]$ on convergence
\State Compute power-sharing ratios $\mu[k]$ using~\eqref{Form-1}
\State Apply power-sharing ratios to BESS

\end{algorithmic}
\label{EnKI-OPM}
\end{algorithm}
\section{Simulation Results}
\begin{table}[!t]
	\renewcommand{\arraystretch}{1.2}
	\caption{Specifications of the BESS}
	\centering
	\label{table_1}
	\resizebox{\columnwidth}{!}{
		\begin{tabular}{l l l}
			\hline\hline \\[-3mm]
			\multicolumn{1}{c}{Symbol} & \multicolumn{1}{c}{Parameter} & \multicolumn{1}{c}{Value [Unit]}  \\[1.6ex] \hline
			$ n $ & Number of battery cells & 100 \\
			$ \bar{Q} $ & Cell nominal capacity & 2.5     [A.h] \\ 
			$ R $ & Cell internal resistance & 30     [m$\Omega$] \\
			$ R_C $ & Cell internal resistance & 10    [m$\Omega$] \\
			$ [q^{\textrm{min}},q^{\textrm{max}}] $ & Cell SoC limits  & [0.05,0.95] \\ 
			$ [i^{\textrm{min}},i^{\textrm{max}}] $ & Cell current limits & [-5,5]     [A] \\ 
			$ C_{\textrm{th}} $ & Thermal capacitance & 40.23     [J/K] \\ 
			$ R_{\textrm{conv}} $ & Convection thermal resistance & 41.05     [K/W] \\ 
			$ T_{\textrm{env}} $ & Environment temperature & 298     [K] \\ 
			$ \Delta q $ & SoC balancing threshold & 1\% \\
			$ \Delta T $ & Temperature balancing threshold & 0.75    [K] \\
            $ \epsilon $ & EnKI convergence tolerance & $10^{-4}$     \\
			$ H $        & Horizon length & 10     [s]\\
			\hline\hline
		\end{tabular}
	}
\end{table}

\begin{figure}[!t]\centering
	\includegraphics[trim={2.2cm 0 2.5cm 0.5cm},clip,width=8cm]{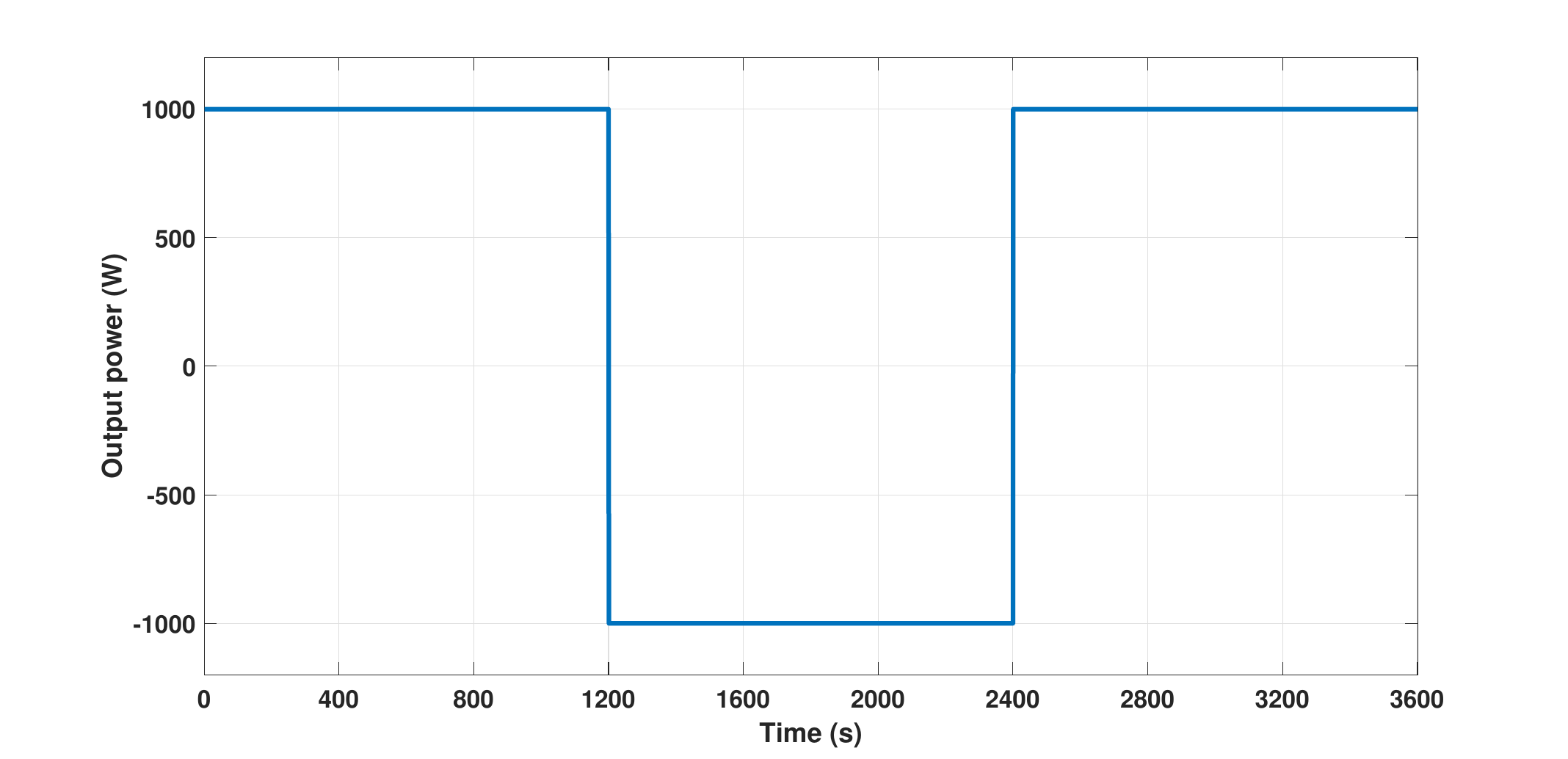}
	\caption{The output power profile.}\label{Fig_SIM_4}
\end{figure}

\begin{figure*}[t]
	    \centering
    \subfloat[\centering ]{{\includegraphics[trim={2.2cm 0 2.5cm 0.5cm},clip,width=8cm]{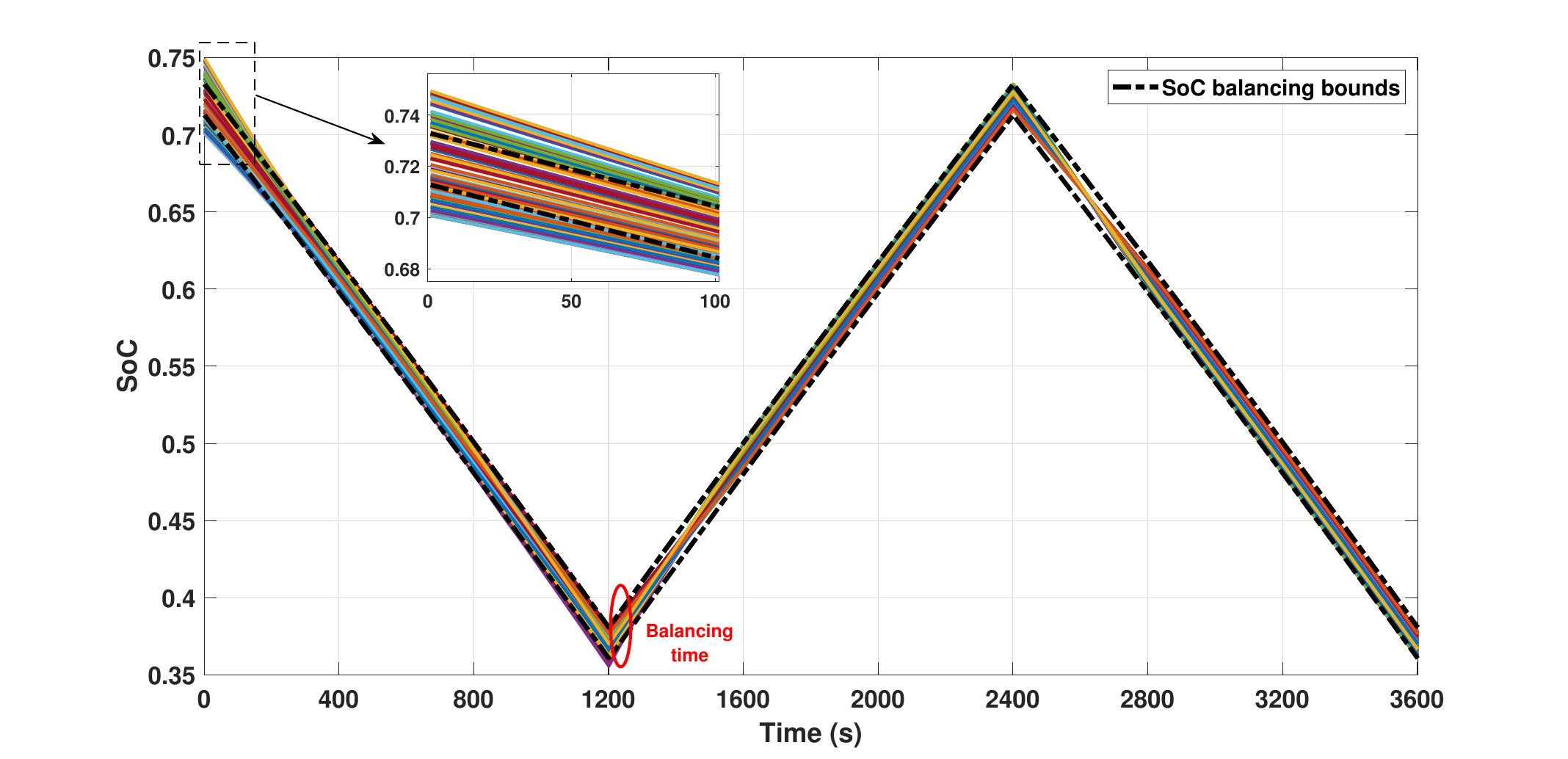} }}
	\,
    \subfloat[\centering ]{{\includegraphics[trim={2.2cm 0 2.5cm 0.5cm},clip,width=8cm]{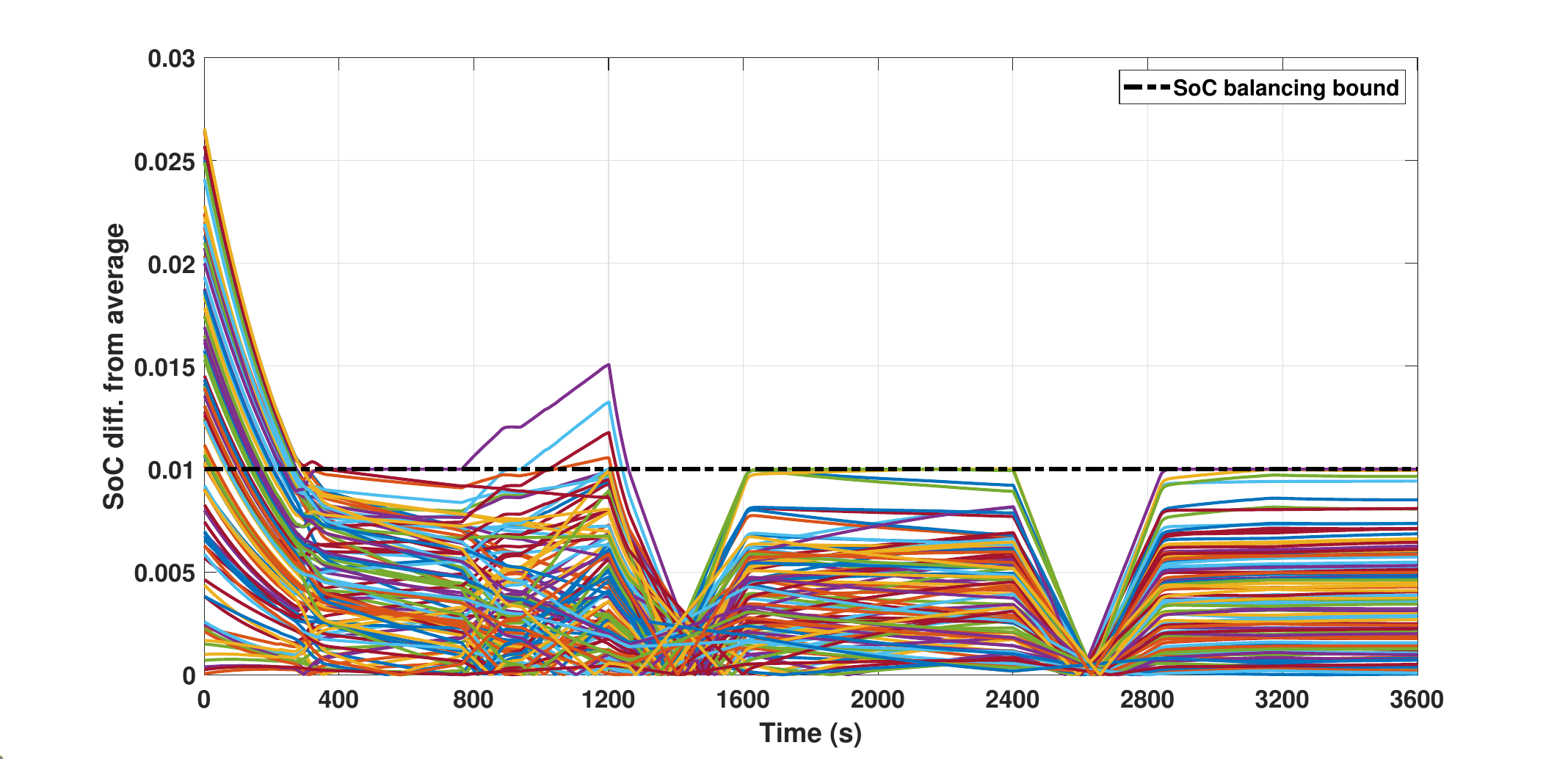} }}
    \,
    \subfloat[\centering ]{{\includegraphics[trim={2.2cm 0 2.5cm 0.5cm},clip,width=8cm]{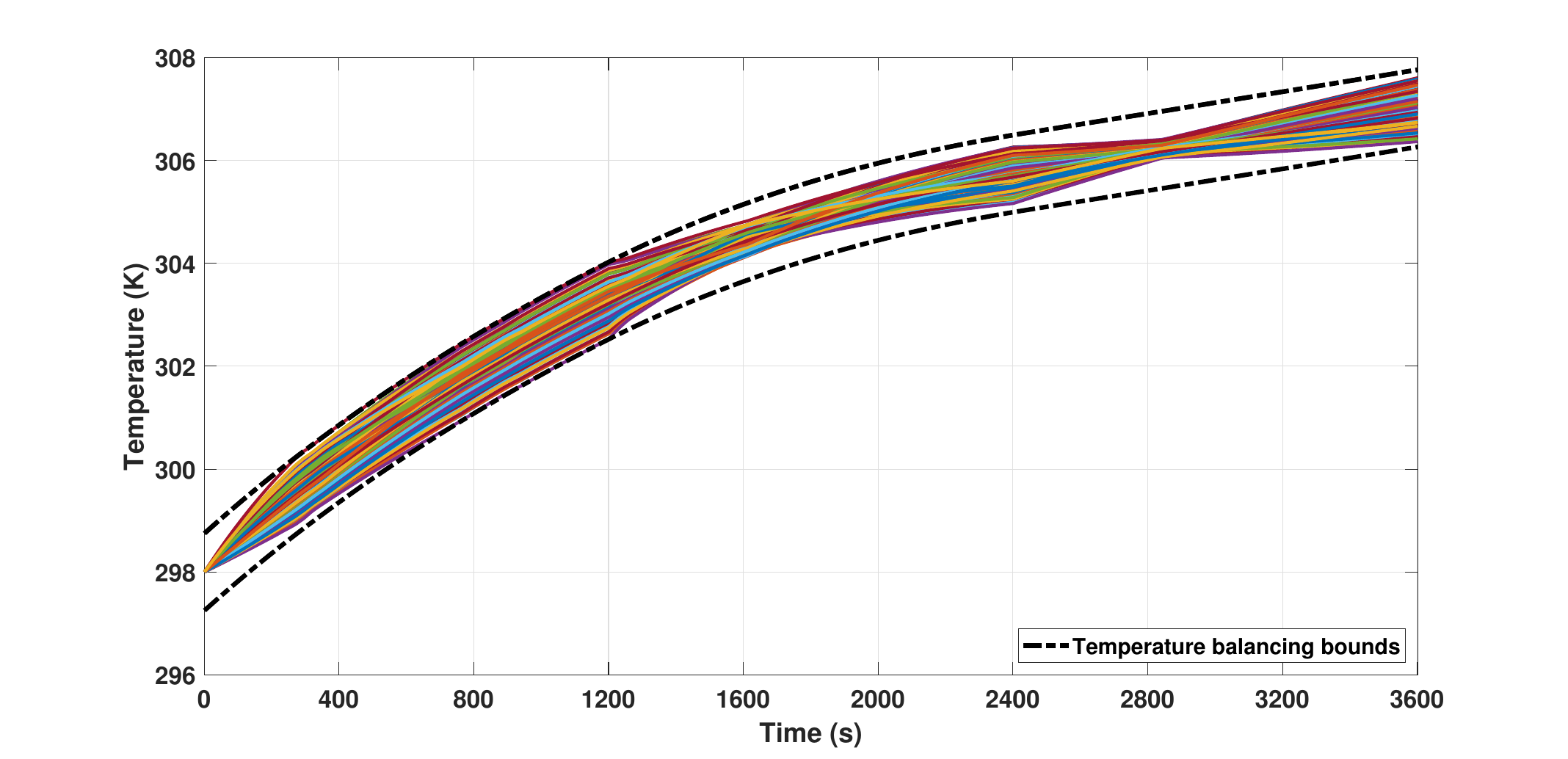} }}
    \,
    \subfloat[\centering ]{{\includegraphics[trim={2.2cm 0 2.5cm 0.5cm},clip,width=8cm]{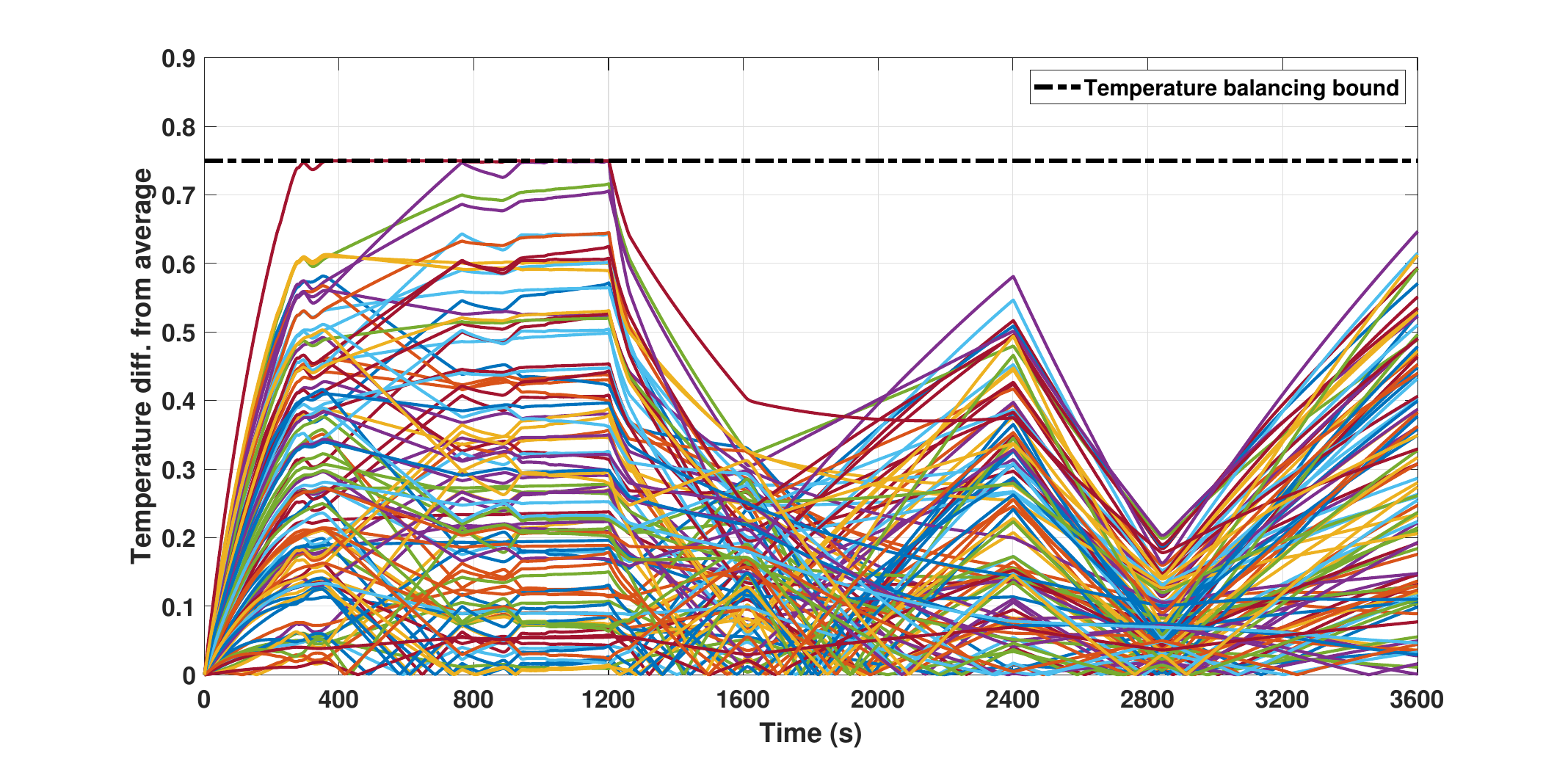} }}
    \caption{Simulation results of the SoC and temperature balancing. (a) The SoC of the cells. (b) The deviation of the cells' SoC from the average. (c) The temperature of the cells. (d) The deviation of the cells' temperature from the average.}
    \label{Fig_SIM_1}
\end{figure*}

This section presents the simulation results of the proposed efficient optimal power management on a BESS comprising 100 cells. The key specifications of the considered BESS are summarized in Table \ref{table_1}, which are based on the Samsung INR18650-25R cells. The output power profile is designed to periodically charge and discharge the BESS at 1 kW every twenty minutes, as illustrated in Fig.~\ref{Fig_SIM_4}. We simulate the BESS operation for one hour with the time steps of one second. We program with Matlab to solve the parameter estimation problem using a workstation with a 3.5 GHz Intel Core i9-10920X CPU and 128 GB of RAM.

For the purpose of implementation, we use a softplus function as the barrier function: 
\begin{equation}
	\psi(x) = \frac{1}{\alpha}\ln(1+\exp(\beta x)),
	\label{Barrier-2}
\end{equation}
where $\alpha$ and $\beta$ are chosen such that $\psi(x)$ is nearly zero when $x \leq 0$ and takes large numbers when $x>0$ \cite{ACC-AI-2021}. Furthermore, $\beta_1$ and $\beta_2$ in \eqref{Form-1} are set to 8 and 12, respectively.

In order to evaluate how well the proposed approach performs, we start by setting the cells' SoC values and internal resistances from uniform distributions $\mathcal{U}(0.7,0.75)$ and $\mathcal{U}(0.03, 0.04)$ $\Omega$, respectively. The cells' initial temperatures are set to be 298 K. The results obtained from these initializations are presented below.

In Fig.~\ref{Fig_SIM_1}, we see the performance of the proposed approach in terms of SoC and temperature balancing. For better interpretation, we also present the estimated parameters in Fig.~\ref{Fig_SIM_2}. The simulation starts with an uneven distribution of SoC across the cells. We observe in Fig.~\ref{Fig_SIM_2} that the proposed approach initially prioritizes SoC balancing by setting $\theta_1$ close to one. This results in the ensuing convergence of the cells' SoC values, as shown in Figs.~\ref{Fig_SIM_1} (a), (b). However, the cells' temperatures start to drift away from each other, and at about 300 seconds, the deviation in temperatures reaches the maximum allowed threshold (see Figs.~\ref{Fig_SIM_1} (c), (d)). At this point, the approach reduces $\theta_1$ and increases $\theta_2$ to maintain the cells' temperature within the desired range, as illustrated in Fig.~\ref{Fig_SIM_2}. 

Moving forward, after around 750 seconds of the simulation, it is no longer possible to maintain both SoC and temperature balanced. This behavior is a consequence of the infeasibility of the underlying optimization problem, where achieving both SoC and temperature balancing is not feasible. As a result, a slight deviation of the cells' SoC is allowed to keep the cells' temperature within the desired bound. This behavior continues until the charging cycle starts for the BESS at the 1200-th second.

The charging cycle allows the proposed approach to attain SoC and temperature balancing simultaneously. Looking at Fig.~\ref{Fig_SIM_2}, we observe that the proposed approach increases $\theta_1$. At roughly the 1300-th second, the cells are balanced, and the proposed approach maintains the balance from thereon. After ensuring that both the SoC and temperature are balanced, the proposed approach focuses on minimizing the power losses by reducing $\theta_1$ and $\theta_2$ after the 1300-second mark. During this period, the BESS operation allocates the output power demand among the cells based on their internal resistances to achieve the power-loss-minimized operation. This period lasts until the 1600-second mark, when the SoC balancing constraint is activated (as shown in Fig.~\ref{Fig_SIM_1} (b)). After this point, the proposed approach increases $\theta_1$ and $\theta_2$ and continues to minimize power losses while also maintaining balanced SoC and temperature among the cells.

At the time instant of 2400 seconds, the BESS transitions to a discharging cycle. This shift allows the proposed approach to enforce power loss minimization while adhering to the SoC and temperature balancing constraints. As a result, $\theta_1$ and $\theta_2$ are decreased, as shown in Fig.~\ref{Fig_SIM_2}. However, the SoC values of the cells diverge until they reach the SoC balancing threshold at around the 2850-th second. At this point, the proposed approach increases $\theta_1$ and $\theta_2$ to keep the cells within the desired balancing bounds.

\begin{figure}[!t]\centering
	\includegraphics[trim={2.2cm 0 2.5cm 0.5cm},clip,width=8cm]{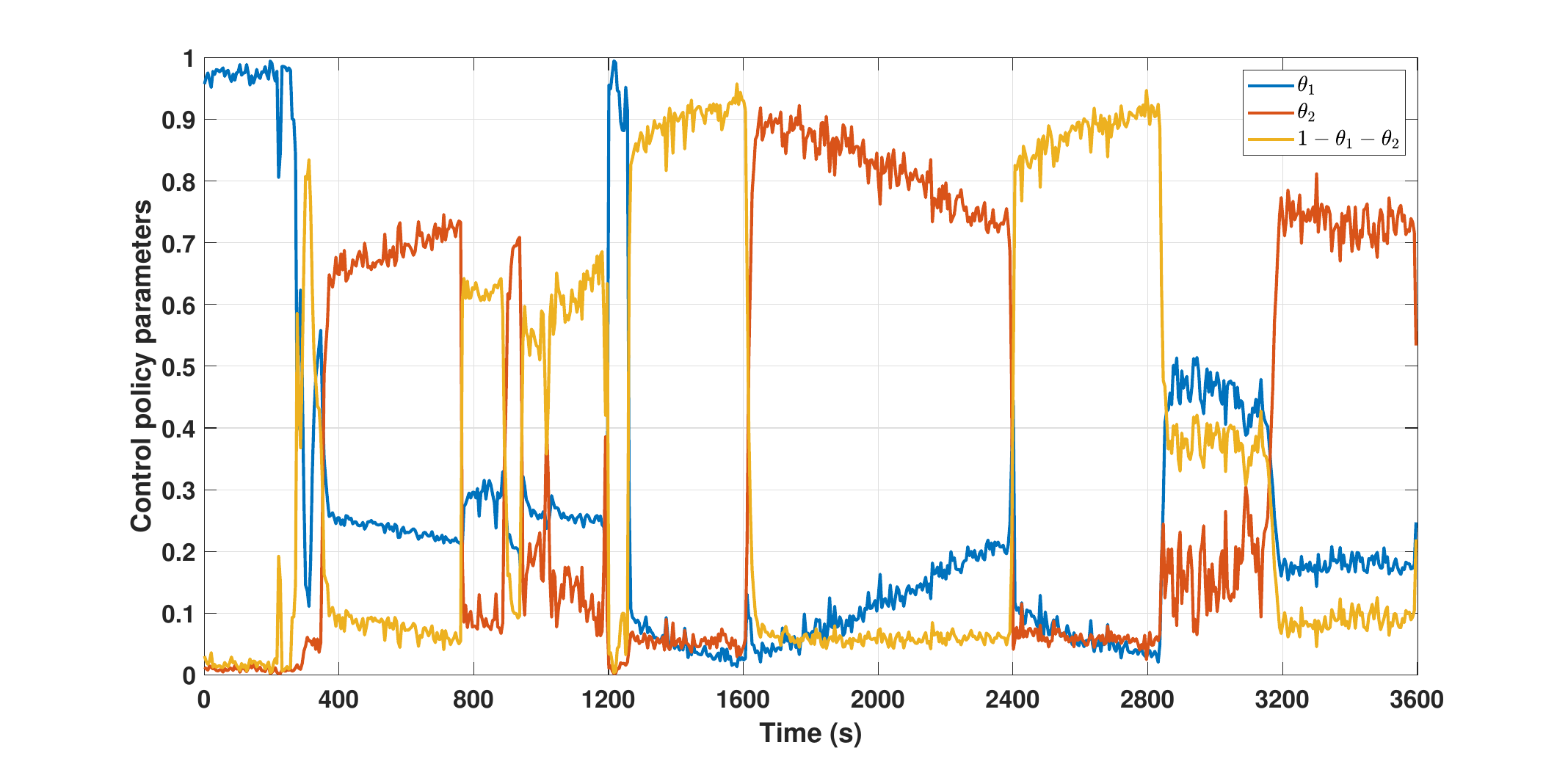}
	\caption{The estimated parameters of the control policy.}\label{Fig_SIM_2}
\end{figure}

Fig.~\ref{Fig_SIM_2} shows the parameters of the designed control policy, and based on these parameters, we use \eqref{Form-1} to calculate the power-sharing ratios of the cells, which are illustrated in Fig.~\ref{Fig_SIM_3}. These power-sharing ratios are used to distribute the output power demand among the cells. This differentiation aims to minimize the power losses while ensuring that the SoC and temperature of each cell are balanced. It is essential to observe that when the cells are almost balanced (e.g., during the time interval of 2900$\sim$3600 seconds), the power-sharing ratios converge towards $1/n$, requiring only minor adjustments to maintain the balancing condition.

\begin{figure}[!t]\centering
	\includegraphics[trim={2.2cm 0 2.5cm 0.5cm},clip,width=8cm]{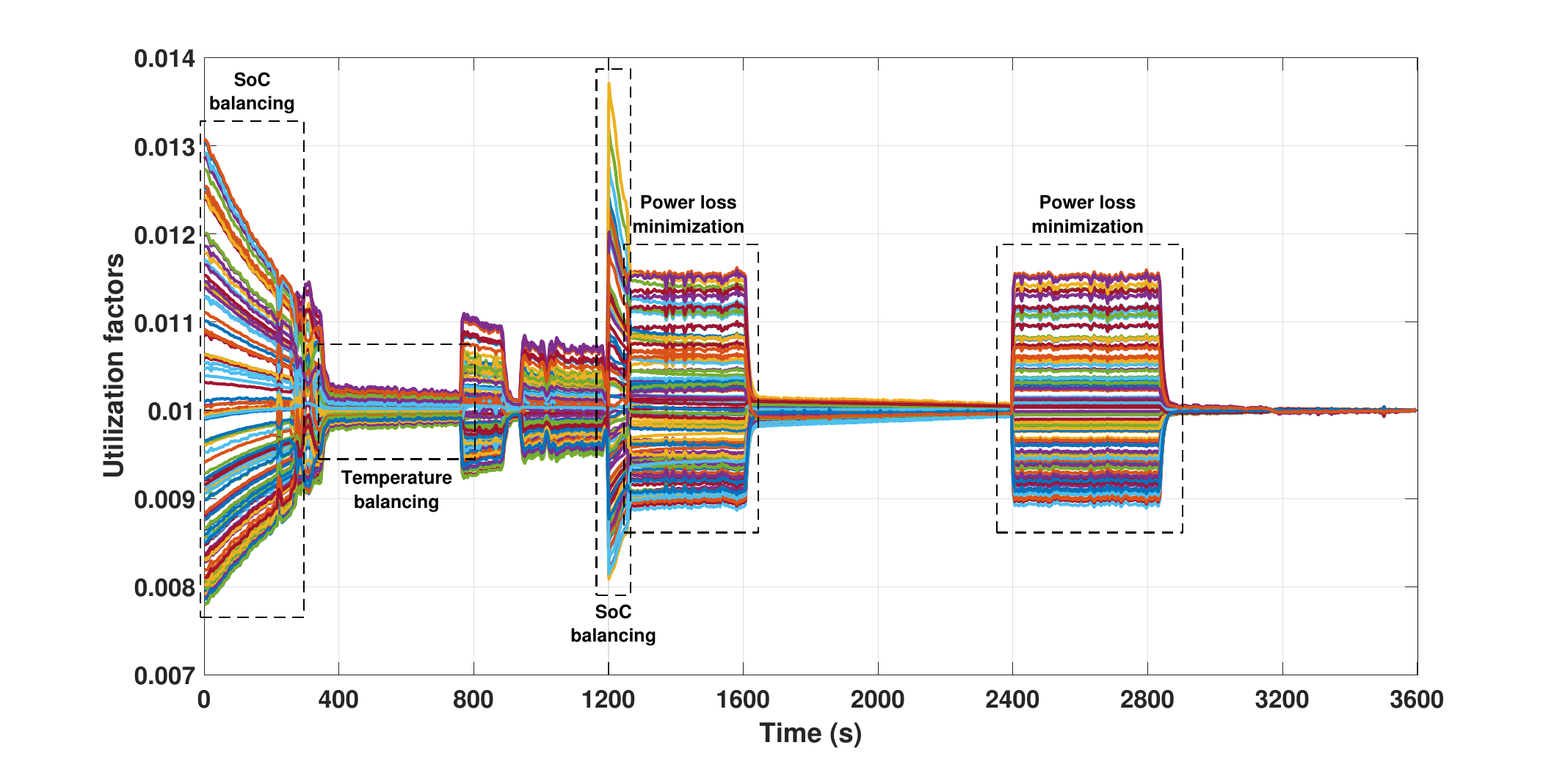}
	\caption{The evolution of the cells' power-sharing ratios.}\label{Fig_SIM_3}
\end{figure}

The objective of our proposed approach is to deliver a computationally efficient solution for the optimal power management of large-scale BESS. Table~\ref{table_2} presents the average computation times associated with our approach, considering different cell numbers and ensemble sizes for the EnKI. Our results demonstrate that the computation time for our proposed approach consistently remains below 10-second, even when handling 100 cells. In contrast, gradient-based numerical optimization \eqref{Optim-1} would typically demand several minutes for just 20 cells and would become impractically burdensome and collapse for 100 cells.

\begin{table}[!t]
	\renewcommand{\arraystretch}{1.2}
	\caption{Computation times of the proposed approach.}
	\centering
	\label{table_2}\
    \small
\begin{tabular}{c||ccc}
\multirow{2}{*}{Ensemble size} & \multicolumn{3}{c}{Number of cells} \\ \cline{2-4} 
                  & \multicolumn{1}{c|}{20} & \multicolumn{1}{c|}{50} & 100 \\ \hline \hline
          EnKI-50 & \multicolumn{1}{l|}{0.79 s} & \multicolumn{1}{l|}{1.33 s} & 3.29 s \\ \hline
          EnKI-100& \multicolumn{1}{l|}{0.99 s} & \multicolumn{1}{l|}{1.89 s} & 6.24 s \\ \hline
          EnKI-200& \multicolumn{1}{l|}{1.57 s} & \multicolumn{1}{l|}{2.89 s} & 6.37 s
\end{tabular}
\end{table}
\section{Conclusions}
As BESS becomes prevalent in various sectors of industry and economy, optimal power management becomes necessary for achieving peak performance. However, optimal power management of large-scale BESS poses significant challenges in terms of computational requirements. To overcome this bottleneck, this paper proposes a computationally efficient approach. We start by introducing power-sharing ratios for the cells, reflecting their corresponding power quota from the output demand. Subsequently, we formulate an optimal power management problem aimed at minimizing power losses, while complying with the safety, balancing, and power conservation constraints. To efficiently solve this problem, we design a parameterized control policy and then recast the optimal power management problem into a parameter estimation task. The ensemble Kalman inversion is employed to perform this parameter estimation. Through extensive simulations, we validate the effectiveness of our proposed approach. Our results demonstrate significant improvements in computation time while maintaining good accuracy.
\balance

\bibliographystyle{IEEEtran}
\bibliography{IEEEabrv,BIB}

\begin{thebibliography}{10}
\providecommand{\url}[1]{#1}
\csname url@samestyle\endcsname
\providecommand{\newblock}{\relax}
\providecommand{\bibinfo}[2]{#2}
\providecommand{\BIBentrySTDinterwordspacing}{\spaceskip=0pt\relax}
\providecommand{\BIBentryALTinterwordstretchfactor}{4}
\providecommand{\BIBentryALTinterwordspacing}{\spaceskip=\fontdimen2\font plus
\BIBentryALTinterwordstretchfactor\fontdimen3\font minus
  \fontdimen4\font\relax}
\providecommand{\BIBforeignlanguage}[2]{{%
\expandafter\ifx\csname l@#1\endcsname\relax
\typeout{** WARNING: IEEEtran.bst: No hyphenation pattern has been}%
\typeout{** loaded for the language `#1'. Using the pattern for}%
\typeout{** the default language instead.}%
\else
\language=\csname l@#1\endcsname
\fi
#2}}
\providecommand{\BIBdecl}{\relax}
\BIBdecl

\bibitem{TIE-AA-2005}
A.~Affanni, A.~Bellini, G.~Franceschini, P.~Guglielmi, and C.~Tassoni,
  ``Battery choice and management for new-generation electric vehicles,''
  \emph{IEEE Transactions on Industrial Electronics}, vol.~52, no.~5, pp.
  1343--1349, 2005.

\bibitem{JPS-DR-2015}
D.~Rosewater and A.~Williams, ``Analyzing system safety in lithium-ion grid
  energy storage,'' \emph{Journal of Power Sources}, vol. 300, pp. 460--471,
  2015.

\bibitem{IECON-AB-2018}
B.~Arabsalmanabadi, N.~Tashakor, A.~Javadi, and K.~Al-Haddad, ``Charging
  techniques in lithium-ion battery charger: Review and new solution,'' in
  \emph{IECON - 44th Annual Conference of the IEEE Industrial Electronics
  Society}, 2018, pp. 5731--5738.

\bibitem{IVPPC-BJV-2014}
J.~V. Barreras, C.~Pinto, R.~de~Castro, E.~Schaltz, S.~J. Andreasen, and R.~E.
  Ara\'{u}jo, ``Multi-objective control of balancing systems for li-ion battery
  packs: A paradigm shift?'' in \emph{IEEE Vehicle Power and Propulsion
  Conference}, 2014, pp. 1--7.

\bibitem{ECC-PM-2013}
M.~Preindl, C.~Danielson, and F.~Borrelli, ``Performance evaluation of battery
  balancing hardware,'' in \emph{European Control Conference}, 2013, pp.
  4065--4070.

\bibitem{ITEC-RG-2015}
R.~Gu, P.~Malysz, M.~Preindl, H.~Yang, and A.~Emadi, ``Linear programming based
  design and analysis of battery pack balancing topologies,'' in \emph{IEEE
  Transportation Electrification Conference and Expo}, 2015, pp. 1--7.

\bibitem{IFAC-NM-2012}
N.~Murgovski, L.~Johannesson, and J.~Sj$\ddot{\textrm{o}}$berg, ``Convex
  modeling of energy buffers in power control applications,'' \emph{IFAC
  Proceedings Volumes}, vol.~45, no.~30, pp. 92--99, 2012, 3rd IFAC Workshop on
  Engine and Powertrain Control, Simulation and Modeling.

\bibitem{ITSE-PC-2016}
C.~Pinto, J.~V. Barreras, E.~Schaltz, and R.~E. Ara\'{u}jo, ``Evaluation of
  advanced control for \uppercase{L}i-ion battery balancing systems using
  convex optimization,'' \emph{IEEE Transactions on Sustainable Energy},
  vol.~7, no.~4, pp. 1703--1717, 2016.

\bibitem{ITVT-CR-2019}
R.~de~Castro, C.~Pinto, J.~Varela~Barreras, R.~E. Ara\'{u}jo, and D.~A. Howey,
  ``Smart and hybrid balancing system: Design, modeling, and experimental
  demonstration,'' \emph{IEEE Transactions on Vehicular Technology}, vol.~68,
  no.~12, pp. 11\,449--11\,461, 2019.

\bibitem{TTE-FA-2023}
A.~Farakhor, D.~Wu, Y.~Wang, and H.~Fang, ``A novel modular, reconfigurable
  battery energy storage system: Design, control, and experimentation,''
  \emph{IEEE Transactions on Transportation Electrification}, vol.~9, no.~2,
  pp. 2878--2890, 2023.

\bibitem{ITCST-CR-2021}
R.~de~Castro, H.~Pereira, R.~E. Ara\'{u}jo, J.~V. Barreras, and H.~C. Pangborn,
  ``q\uppercase{TSL}: A multilayer control framework for managing capacity,
  temperature, stress, and losses in hybrid balancing systems,'' \emph{IEEE
  Transactions on Control Systems Technology}, pp. 1--16, 2021.

\bibitem{TPEL-MT-2016}
T.~Morstyn, M.~Momayyezan, B.~Hredzak, and V.~G. Agelidis, ``Distributed
  control for state-of-charge balancing between the modules of a reconfigurable
  battery energy storage system,'' \emph{IEEE Transactions on Power
  Electronics}, vol.~31, no.~11, pp. 7986--7995, 2016.

\bibitem{ITSE-OQ-2018}
Q.~Ouyang, J.~Chen, J.~Zheng, and H.~Fang, ``Optimal cell-to-cell balancing
  topology design for serially connected lithium-ion battery packs,''
  \emph{IEEE Transactions on Sustainable Energy}, vol.~9, no.~1, pp. 350--360,
  2018.

\bibitem{ACC-FA-2023}
A.~Farakhor, Y.~Wang, D.~Wu, and H.~Fang, ``Distributed optimal power
  management for battery energy storage systems: A novel accelerated tracking
  admm approach,'' in \emph{American Control Conference}, 2023, pp. 3106--3112.

\bibitem{ENE-HH-2011}
H.~He, R.~Xiong, and J.~Fan, ``Evaluation of lithium-ion battery equivalent
  circuit models for state of charge estimation by an experimental approach,''
  \emph{Energies}, vol.~4, no.~4, pp. 582--598, 2011.

\bibitem{IP-NK-2019}
N.~B. Kovachki and A.~M. Stuart, ``Ensemble {K}alman inversion: a
  derivative-free technique for machine learning tasks,'' \emph{Inverse
  Problems}, vol.~35, no.~9, p. 095005, 2019.

\bibitem{MWR-PH-2016}
P.~L. Houtekamer and F.~Zhang, ``Review of the ensemble kalman filter for
  atmospheric data assimilation,'' \emph{Monthly Weather Review}, vol. 144,
  no.~12, pp. 4489 -- 4532, 2016.

\bibitem{StatLetters-SD-2022}
S.~Duffield and S.~S. Singh, ``Ensemble {K}alman inversion for general
  likelihoods,'' \emph{Statistics \& Probability Letters}, vol. 187, p. 109523,
  2022.

\bibitem{ACC-AI-2021}
I.~Askari, S.~Zeng, and H.~Fang, ``Nonlinear model predictive control based on
  constraint-aware particle filtering/smoothing,'' in \emph{American Control
  Conference}, 2021, pp. 3532--3537.

\end{thebibliography}

\end{document}